\documentclass[amsmath,11pt]{article}
\usepackage{times}
\usepackage{fullpage}
\usepackage{hyperref,microtype,pdfsync}
\usepackage{amsmath,amsfonts,amssymb,amsthm}
\usepackage{nicefrac}
\usepackage{url}
\usepackage{enumitem}
\usepackage{algorithm}
\usepackage[noend]{algorithmic}

\newif\ifnotes\notesfalse

\ifnotes

\newcommand{\dnote}[1]{{\bf (Daniel:} {#1}{\bf ) }}
\newcommand{\snote}[1]{{\bf (Santosh:} {#1}{\bf ) }}

\else

\newcommand{\dnote}[1]{}
\newcommand{\snote}[1]{}

\fi

\theoremstyle{plain}            
\newtheorem{theorem}{Theorem}[section]
\newtheorem{lemma}[theorem]{Lemma}
\newtheorem{corollary}[theorem]{Corollary}

\theoremstyle{definition}       

\theoremstyle{remark}           

\numberwithin{equation}{section}

\DeclareMathOperator{\poly}{poly}

\DeclareMathOperator*{\E}{E}

\newcommand{\R}{\ensuremath{\mathbb{R}}}

\newcommand{\Z}{\ensuremath{\mathbb{Z}}}

\newcommand{\pr}[2]{\langle{#1, #2}\rangle}

\newcommand{\abs}[1]{\lvert{#1}\rvert}

\newcommand{\set}[1]{\{{#1}\}}

\newcommand{\round}[1]{\lfloor{#1}\rceil}

\newcommand{\length}[1]{\lVert{#1}\rVert}

\def\eps{\epsilon}

\def\vol{\mathrm{vol}}

\title{Deterministic Construction of an Approximate M-Ellipsoid and its Application to Derandomizing Lattice Algorithms}

\author{Daniel Dadush\thanks{School of Industrial and Systems Engineering, Georgia Tech. {\tt dndadush@gmail.com}}
  \and
  Santosh Vempala\thanks{School of Computer Science, Georgia Tech. {\tt vempala@gatech.edu}}
}

\begin{document}

\maketitle

\begin{abstract}
We give a deterministic $O(\log n)^n$ algorithm for the {\em Shortest Vector Problem (SVP)} of a lattice under {\em any}
norm, improving on the previous best deterministic bound of $n^{O(n)}$ for general norms and nearly matching the bound
of $2^{O(n)}$ for the standard Euclidean norm established by Micciancio and Voulgaris (STOC 2010). Our algorithm can be
viewed as a derandomization of the AKS randomized sieve algorithm, which can be used to solve SVP for any norm in
$2^{O(n)}$ time with high probability. We use the technique of covering a convex body by ellipsoids, as introduced for
lattice problems in (Dadush et al., FOCS 2011).

Our main contribution is a deterministic approximation of an M-ellipsoid of any convex body.
We achieve this via a convex programming formulation of the optimal ellipsoid
with the objective function being an $n$-dimensional integral that we show can be approximated
deterministically, a technique that appears to be of independent interest.
  \ifnotes \begin{center}{\Huge{NOTES ARE ON}}\end{center} \fi
\end{abstract}

\thispagestyle{empty}

\newpage

\setcounter{page}{1}

\section{Introduction}

The Shortest Vector Problems (SVP) on lattices is a central algorithmic problems in the geometry of numbers, with
applications to Integer Programming~\cite{lenstra83:_integ_progr_with_fixed_number_of_variab}, factoring polynomials
over the rationals~\cite{lenstra82:_factor}, cryptanalysis
(e.g.,~\cite{odlyzko90:_rise_and_fall_of_knaps_crypt,DBLP:journals/joc/JouxS98,DBLP:conf/calc/NguyenS01}), and much
more.  (An $n$-dimensional \emph{lattice} $L$ is a discrete additive subgroup of $\R^{n}$, and is generated as the set
of integer linear combinations of some basis vectors $b_{1}, \ldots, b_{k} \in \R^{n}$, for some $k \leq n$.)  The SVP
is simply: given a lattice $L$ represented by a basis, find a nonzero $v \in L$ such that $\length{v}$ is minimized,
where $\length{\cdot}$ denotes a particular norm on $\R^{n}$.  

The fastest known algorithms for solving SVP in general norms, are $2^{O(n)}$ time algorithms based on the AKS Sieve
\cite{DBLP:conf/stoc/AjtaiKS01, DBLP:conf/fsttcs/ArvindJ08}. These algorithms use an exponential amount of randomness
and only guarantee the correctness of their outputs with high probability. Improving on this, \cite{DPV-SVP-11} gave a
$2^{O(n)}$ Las Vegas algorithm (i.e. only the runtime is random, not the correctness) for general norm SVP which uses
only a polynomial amount of randomness. In this paper, building on the ideas of \cite{DPV-SVP-11}, we give a deterministic
$O(\log n)^n$ algorithm for general norm SVP, hence completely eliminating the randomness while sustaining a moderate
slowdown in the running time. The previous best deterministic complexity for general norm SVP is $n^{\Omega(n)}$. 

We review the ideas behind \cite{DPV-SVP-11}. For the Euclidean norm (when $K$ is a ball in $\R^n$), 
Micciancio and Voulgaris \cite{DBLP:conf/stoc/MicciancioV10} showed how to solve the SVP in time $2^{O(n)}$, using a new enumeration technique
based on using the voronoi cell of a lattice (the set of points in $\R^n$ closer to the origin than any other lattice
point). Unfortunately, the direct generalization of their technique to other norms (i.e., using the associated voronoi
cell of the norm), even for $\ell_p$ norms, seems to break down. 

In \cite{DPV-SVP-11}, Dadush et al. proposed a different approach that uses the enumeration
technique \cite{DBLP:conf/stoc/MicciancioV10} and directly reduces SVP in general norms to enumeration in the $\ell_2$
norm. Their key idea was to use the classical $M$-ellipsoid covering from convex geometry to cover a given convex body
$K$ by a small number of ellipsoids each of roughly the same volume as $K$. An $M$-ellipsoid of a convex body $K$ is an
ellipsoid $E$ with the following properties:
\begin{enumerate}
\item $N(K,E) \le 2^{O(n)}$
\item $N(E,K) \le 2^{O(n)}$
\end{enumerate}
where $N(A,B) = \inf \set{|\Lambda|: \Lambda \subseteq \R^n, A \subseteq B + \Lambda}$ is the number of translations of
$B$ required to cover $A$. In words, the number of copies of $E$ required to cover $K$ and vice versa are both bounded
by a single exponential in $n$. The existence of such an ellipsoid for any convex body was established by Milman
\cite{M86}. We note that an M-ellipsoid can be quite different from the more classical John ellipsoid, e.g. the largest
ellipsoid contained in $K$, since its volume can be an $n^{O(n)}$ factor off from $K$ (e.g., the cube vs the unit ball)
implying than $N(K,E) = n^{\Omega(n)}$. 

The first step in \cite{DPV-SVP-11} is reduce to general norm SVP under $\|\cdot\|_K$ and a lattice $L$ to lattice point
enumeration inside a scaling of $K$, in particular any scaling $s > 0$ such that $sK \cap L \neq \emptyset$ and
$\frac{s}{2}K \cap L = \emptyset$ (which can easily be guessed). Importantly, at this scaling, it is shown that $sK$
never contains more than $2^{O(n)}$ lattice points in any translation. The main idea in \cite{DPV-SVP-11} is then that
enumerating the lattice points inside $sK$ reduces to enumerating the lattice points inside the ellipsoids in an
M-ellipsoid covering of $sK$, thereby reducing the problem to enumeration in $\ell_2$ (which can be solved using the
techniques in \cite{DBLP:conf/stoc/MicciancioV10}). Given the covering properties satisfied by the M-ellipsoid, we 
get that the total number of lattice points enumerated in this way is at most a $2^{O(n)}$ factor more than the maximum
number of lattice points $K$ can contain in any translation, and hence $2^{O(n)}$. 

Thus a key ingredient in the approach of \cite{DPV-SVP-11} to solve SVP under $\|\cdot\|_K$ is finding an
M-ellipsoid of $K$. Indeed, the paper \cite{DPV-SVP-11} gives a polynomial-time {\em randomized}
algorithm to construct an M-ellipsoid with high probability, based on the techniques of Klartag
\cite{K06} (such an algorithm was implicit in his paper). Unfortunately, the algorithm makes essential 
use of random sampling over convex bodies and seems inherently difficult to derandomize.

In this paper, we give a deterministic algorithm to build an ``approximate'' M-ellipsoid $E$ for any convex body $K$.
While we do not obtain the optimal covering bounds, we will guarantee that  $N(K,E) = 2^{O(n)}$ 
and $N(E,K) = O(\log n)^n = 2^{O(n\log\log n)}$. Moreover, we show that this ellipsoid $E$ can be computed $O(\sqrt{\log n})^n$ time. This
result and its consequence for the SVP are stated more precisely in the following theorems.

\begin{theorem} \label{thm:M-ellipsoid-algorithm}
There is deterministic $O(\log n)^n$-time algorithm that given any convex body $K \subset \R^n$, specified by
a membership oracle, finds an ellipsoid $E$ such that $N(K,E) \le 2^{O(n)}$ and $N(E,K) \le O(\log n)^n$. The complexity
of the algorithm (oracle calls and arithmetic operations) is $O(\sqrt{\log n})^n$. 
\end{theorem}

Using this theorem, and the techniques from \cite{DPV-SVP-11}, we obtain the following result:

\begin{theorem}\label{thm:det-SVP-anynorm}
Given a lattice $L$ by a basis and a norm $\|.\|_K$ specified by a convex body $K$, the shortest vector in $L$ under the
norm $\|.\|_K$ can be found in time $O(\log n)^n$. 
\end{theorem}

Applications to other lattice problems (closest vector, integer programming) are described in Section \ref{sec:SVP}.
These results are based on two main ideas. The first is a convex program inspired by an existential approximation to the
$M$-ellipsoid based on a position called the $\ell$-position, given by Pisier \cite{Pis89}. The second is an algorithm
for solving the convex program, where the key hurdle is an efficient deterministic approximation of the objective value at any given
feasible point. 

In the next section, we describe the $\ell$-position which leads to the approximate $M$-ellipsoid. Then we give our
convex programming based algorithm for computing the approximate $M$-ellipsoid, followed by its analysis. Section
\ref{sec:SVP} applies this to the SVP and other problems. 

We conclude this section with a comment on the complexity of computing (approximate) $M$-ellipsoids (and therefore the $\ell$-position).
An $M$-ellipsoid $E$ for a convex body $K$ achieving covering numbers $N(K,E)$, $N(E,K)$ gives an $N(K,E)N(E,K)$ to the volume of 
$K$. It is well-known that in the oracle model for convex bodies, any deterministic algorithm that has complexity at most $n^a$ incurs an approximation 
factor of $(cn/a\log n)^{n/2}$, implying in particular that an algorithm that achieves a $2^{O(n)}$ approximation must have complexity
$2^{\Omega(n)}$. Theorem \ref{thm:M-ellipsoid-algorithm} readily implies an $O(\log n)^n$ approximation with
 $O(\sqrt{\log n})^n$ complexity, getting close to the lower bound. Fully closing this gap is an interesting open problem.

\section{$M$-ellipsoids and the $\ell$-position}

As explained above, one useful view of whether an ellipsoid $E$ ``approximates'' a convex body $K$ well is if
$N(K,E),N(E,K) = 2^{O(n)}$. A similar view, taken by Pisier, is to find an ellipsoid $E$ with the property that 
$\vol(K \cap E) \ge \vol(E)/2$ and $\vol(K)$ not much larger than $\vol(E)$. 

This is useful in light of the following elementary bound on covering numbers for centrally symmetric bodies (see
\cite{MP00}).
\begin{lemma}
\label{lem:cov-est}
Let $A,B \subseteq \R^n$ be symmetric convex bodies. Then
\[
N(A,B) \leq 3^n \frac{\vol(A)}{\vol(A \cap B)}
\]
\end{lemma}

We are now ready for the $\ell$-position which lets us find an ellipsoid with small covering numbers using this perspective.

Let $K \subseteq \R^n$ be a symmetric convex body, and let $K^* = \set{x: \sup_{y \in K} \pr{x}{y} \leq 1}$ denote the
polar of $K$. Let $B_2^n \subseteq \R^n$ denote the unit euclidean ball, and $S^{n-1} = \partial B_2^n$ denote the unit
sphere.  Let $\gamma_n(x) = \left(\frac{1}{\sqrt{2 \pi}}\right)^n e^{-\frac{1}{2}\|x\|^2}$ be the density of the
canonical gaussian measure on $\R^n$. We define the expected norm of a random Gaussian point as
\[
\ell(K) = \int \|x\|_K \gamma_n(x) dx.
\]

The following lemma, see $\cite{Pis89}$, provides an asymptotic estimate of this quantity.
\begin{lemma}
\label{lem:l-pos-ess}
Let $K \subseteq \R^n$ be a symmetric convex body. Then for
\[
m = \sup \set{r \geq 0: \vol_{n-1}(rS^{n-1}  \cap K) \geq \frac{1}{2} \vol_{n-1}(rS^{n-1})}
\]
we have that $l(K) = \Theta\left(\frac{\sqrt{n}}{m}\right)$.
Furthermore, $\vol(mB_2^n \cap K) \geq \frac{1}{2}\vol(mB_2^n)$.
\end{lemma}

A theorem of Pisier \cite{Pis89} relates the $\ell$-estimate of a body with that of its dual. 
\begin{theorem}
\label{thm:pis-l-pos}
Let $K \subseteq \R^n$ be a symmetric convex body. Then
\[
\inf_{T \in SL(n)} l(TK) l(T^*K^*) \leq c n \log n
\]
where $SL(n)$ is the set of $n \times n$ matrices of determinant $1$ and $c > 0$ is an absolute constant.
\end{theorem}

The next theorem, known as the Blashke-Santal{\'o} inequality \cite{Bla18, San49}, gives an upper bound on the volume
product, a fundamental quantity in convex geometry.
\begin{theorem}[Blashke-Santal{\'o}]
\label{thm:bl-san}
Let $K \subseteq \R^n$ be a symmetric convex body. Then
\[
\vol(K) \vol(K^*) \leq \vol(B_2^n)^2
\]
with equality iff $K$ is an ellipsoid.
\end{theorem}

Using the above estimates, we get the following well-known result, whose proof we include for completeness.
\begin{theorem}[Pisier]
\label{thm:l-pos-cov}
Let $K \subseteq \R^n$ be a symmetric convex body. Then there exists
an ellipsoid $E \subseteq \R^n$ such that
\[
\vol(E \cap K) \geq \frac{1}{2} \vol(E) \quad \text{ and } \quad \vol(K) \leq O(\log n)^n \vol(E \cap K)
\]
In addition, we get that
\[
N(K,E) = O(\log n)^n \quad \text{ and } \quad N(E,K) = \frac{1}{2} 3^n
\]
\end{theorem}
\begin{proof}
Let us first apply a measure preserving linear transformation $T$ to $K$ such that $l(TK)l(T^*K^*)$ is minimized, and
hence by~\ref{thm:pis-l-pos} we may assume that $l(K)l(K^*) = O(n \log n)$. Now using Lemma \ref{lem:l-pos-ess} we see
that
\[
m = \sup \set{r \geq 0: \vol(rB_2^n \cap K) \geq \frac{1}{2} \vol(rB_2^n)} = \Omega\left(\frac{\sqrt{n}}{l(K)}\right)
\]
and that
\[
m^* = \sup \set{r \geq 0: \vol(rB_2^n \cap K^*) \geq \frac{1}{2} \vol(rB_2^n)} = \Omega\left(\frac{\sqrt{n}}{{l(K^*)}}\right)
\]
Hence we get that
\[
m m^* = \Omega\left(\frac{1}{\log n}\right)
\]
Using Theorem~\ref{thm:bl-san} we get that
\begin{align*}
\vol(K) &\leq \frac{\vol(B_2^n)^2}{\vol(K^*)} \leq 2 \frac{\vol(B_2^n)^2}{\vol(m^* B_2^n)}
        = \left(\frac{1}{m^*}\right)^n \vol(B_2^n) \\
        &= O(m \log n)^n \vol(B_2^n) = O(\log n)^n \vol(m B_2^n) = O(\log n)^n \vol(m B_2^n \cap K)
\end{align*}
We now see that the ellipsoid $E = mB_2^n$ satisfies the claims of the corollary. To derive the
additional assertions, we simply apply Lemma \ref{lem:cov-est} to the volume estimates above.
\end{proof}

\section{Algorithm to compute an $\ell$-type Ellipsoid}

Our algorithm will find an ellipsoid by (approximately) solving the following convex program (CP).

\begin{align}
\label{ell-prog}
\begin{split}
\inf f(A) &= \int_{\R^n} \|Ax\|_{K} \gamma_n(x)dx \\
\text{subject to} & \\
	           \quad A &\succeq 0 \\
                   \quad \det(A) & \geq 1
\end{split}
\end{align}

The above program models a tractable formulation of the implicit optimization problem in Theorem \ref{thm:pis-l-pos}.
Indeed it is not hard to show that the $\ell$-ellipsoid (understood by its associated linear transformation) alluded to
in Theorem \ref{thm:pis-l-pos}, in fact gives a feasible solution to the above program of good quality. Hence the optimal
solution to the above program, will be at least as good as the $\ell$-ellipsoid for our purposes. Hence to
yield our approximate M-ellipsoid, it suffices to solve the above program.

In the above program, $K$ will be a symmetric convex body presented by a weak membership oracle, satisfying
$rB_2^n \subseteq K \subseteq RB_2^n$. To solve the program, we first round $K$ using the ellipsoid method \cite{GLS} so
that $B_2^n \subseteq K \subseteq n B_2^n$ (note the improvement from $n^{\frac{3}{2}}$ to $n$ is possible since $K$ is
centrally symmetric). Next we use a discrete approximation of space to approximate the $\ell$-estimate at any given
$A$, where this approximation remains convex. Next we analyze the properties of the above convex program, showing that (1) a
well sandwiched subset of the feasible region (ratio of inner contained and outer containing ball) contains the optimal
solution, (2) the objective function is Lipshitz, and (3) the objective value of the optimal solution is not too small.
From here, we apply the classical reduction from weak membership to weak optimization \cite{GLS} (which simulates the
ellipsoid method), which allows us to compute a $(1+\eps)$ approximation (multiplicative) of the optimal solution using
at most a polynomial number of queries to the objective function. 

Our approximation of the $\ell$-estimate is as follows:
Let 
\[
s = \frac{1}{\sqrt{2\pi}} \sqrt{\frac{\log(2(2n+1))}{\pi}}, \quad  
C_s = \frac{1}{2s}[-1,1]^n \quad \mbox{ and } 
p_x = \int_{C_s} \gamma_n(x + y)dy.
\]
Define $D \subseteq \R^n$ be set of points from the lattice $(1/s)\Z^n$ that lie 
in the ball of radius $3\sqrt{n}$ around the origin, i.e.,
\[
D = \left(\frac{1}{s} \Z^n\right) \bigcap \left(3\sqrt{n}B_2^n\right)
\]
Then 
\[
\tilde{f}(A)=\sum_{x \in D} p_x \|Ax\|_K.
\]

We conclude the description of the algorithm by bounding the size of $D$ and observing that it can be efficiently enumerated. 

First we note that $|D| = O(\sqrt{\log n})^n$. Since $C_s$ tiles space with respect to $\frac{1}{s}
\Z^n$ and $C_s \subseteq \sqrt{n}B_2^n$, we have that
\[
|D| = \frac{\vol(D + C_s)}{\vol(C_s)} \leq \frac{\vol(3\sqrt{n}B_2^n + C_s)}{\vol(C_s)}
    \leq \frac{\vol(4\sqrt{n}B_2^n)}{s^{-n}} = 4^n \vol(\sqrt{n}B_2^n) s^n = O(\sqrt{\log n})^n
\]
as claimed. 

It is straightforward to compute the set $D$ using $O(\sqrt{\log n})^n$ time and space. To see
this, we observe that the graph induced on $D$ by connecting elements $x,y \in D$ iff $x-y \in \pm \frac{1}{s}
\set{e_1,\dots,e_n}$ is connected: a path to the origin can be constructed from any $v \in D$ by decreasing each
component of $v$ by $\frac{1}{s}$ until it hits zero. Hence a breadth-first or depth-first search 
of this graph starting from the origin allows us to compute all of $D$ in the required time.

\section{Analysis}

The analysis is divided into two parts. First, we give an $O(\sqrt{\log n})^n$ algorithm to compute an approximation of
the objective value in \ref{ell-prog} on any given input. Second, we show that the optimization problem with the
approximated objective \ref{ell-prog} is well-behaved, i.e. that it is convex, that the feasible region can be nicely
bounded, the objective function is Lipshitz. This will allow us to apply the ellipsoid algorithm to solve the problem.

\subsection{Computing the $\ell$-estimate}

In this section, we analyze the deterministic algorithm to approximately compute $\ell(K)$ in $O(\sqrt{\log n})^n$ time.
Recall that our approach is to approximate the associated integral as a sum over a discrete set.

We first describe the idea. A reasonable first approach would be to check whether the integrand (i.e. $\|x\|_K$) is
Lipschitz enough so that reasonably sized discretization may be used to approximate the integral
$\ell(K)$. Indeed, it will be true that $|\|x\|_K-\|y\|_K| \leq O(\ell(K))\|x-y\|_2$. Given that
the mass of the $n$ dimensional standard gaussian is concentrated inside of shell of constant width at radius
$\sqrt{n}$, this bound on the Lipshitz constant would suggest that a discretization $D$ of $\sqrt{n}S^{n-1}$, such
that every point in $\sqrt{n}S^{n-1}$ is at distance $O(1)$ from $D$, should suffice to estimate $\ell(K)$. 
Though this will indeed be true, any such discrete set $D$ must have size $O(\sqrt{n})^n$, i.e. far lager than $O(\sqrt{\log
n})^n$. Taking a closer look however, we observe that one only needs such a Lipschitz bound ``on average", since all we
want is to approximate is the integral. This we are able to bound below, using some standard tail bounds and a
simple monotonicity inequality about expectations. 

To perform the analysis of our algorithm, 
we will need certain facts about the discrete Gaussian distribution.  Let 
\[
\rho_s(x) = e^{-\pi\|\frac{x}{s}\|^2}
\]
for $x \in \R^n$, and we write $\rho_s(A)$ to mean $\sum_{x \in A} \rho_s(x)$
for $A \subseteq \R^n$. 
For an $n$-dimensional lattice $L \subseteq \R^n$, and $c \in \R^n$ we define the discrete Gaussian
measure on $L+c$ with parameter $s$ as
\[
D_{L+c,s}(A) = \frac{\rho_s(A)}{\rho_s(L+c)}
\]
for $A \subseteq L+c$.

In our setting, we will only need the case $L = \Z^n$. 
We let $U$ stand for the uniform distribution on $[-1/2,1/2]^n$.
We now state some useful standard lemmas. See \cite{DBLP:journals/dcg/Banaszczyk95, DBLP:journals/siamcomp/MicciancioR07}.


\begin{lemma}
\label{lem:zn-ineqs}
Take $s \geq \sqrt{\frac{\log(2(t+1))}{\pi}}$ and let $X$ be distributed as $D_{L+c,s}$ for
$c \in \R^n$. Then
\[
\left(1-\frac{1}{t}\right)^n s^n \leq \rho_s(\Z^n + c) \leq \left(1+\frac{1}{t}\right)^n s^n
\]
\end{lemma}

\begin{lemma}
\label{lem:gaussian-tail}
Let $X$ be drawn from a standard $n$-dimensional Gaussian $N(0,1)^n$, i.e., 
with density $\left(\frac{1}{\sqrt{2\pi}}\right)^n e^{-\frac{1}{2}\|x\|^2}$,
then for $t \geq 1$ we have that
\[
\Pr(\|X\| \geq t\sqrt{n}) \leq  e^{-\left(1-\frac{1+\ln(t^2)}{t^2}\right) \frac{1}{2} n t^2}
\]
\end{lemma}

The next lemma is an inequality that we will use in the main proof. 

\begin{lemma}
\label{lem:unif-gaus-comp}
Let $f:\R^n \rightarrow \R$ be a convex function. Let $U$ denote the uniform distribution on
$[-\frac{1}{2},\frac{1}{2}]^n$ and let $X$ denote the $n$-dimensional Gaussian $N(0,1/\sqrt{2\pi})$, 
i.e., with density $e^{-\pi\|x\|^2}$. Then we
have that
\[
\E[f(X)] \geq \E[f(U)]
\]
\end{lemma}
\begin{proof}
We shall prove the statement by induction. Let $C = [-\frac{1}{2},\frac{1}{2}]$. We start with the base case $n=1$.  The
density of $U$ here is $I[x \in C]$, and the density for $X$ is $e^{-\pi x^2}$ (this density function 
is chosen so that the density is at most $1$ everywhere).

For our convex function $f:\R \rightarrow \R$, let $\phi$ denote the linear function satisfying $\phi(-\frac{1}{2}) =
f(-\frac{1}{2})$ and $\phi(\frac{1}{2}) = f(\frac{1}{2})$. By convexity of $f$ we note that $f(x) \leq \phi(x)$ for $x
\in C$ and $f(x) \geq \phi(x)$ for $x \in \R \setminus C$. Now we note that
\[
\E[f(X)] - \E[f(U)] = \int_\R f(x)(e^{-\pi x^2} - I[x \in C]) dx = \int_{\R \setminus C} f(x)(e^{-\pi x^2})dx
+ \int_C f(x)(e^{-\pi x^2} - 1)dx
\]
For $x \in \R \setminus C$, we have that $e^{-\pi x^2} \geq 0$ and $f(x) \geq \phi(x)$, and hence
\[
\int_{\R \setminus C} f(x)(e^{-\pi x^2}) \geq \int_{\R \setminus C} \phi(x)(e^{-\pi x^2}) \text{.}
\]
For $x \in C$, we have that $e^{-\pi x^2} \leq 1$ and that $f(x) \leq \phi(x)$, and hence
\[
\int_C f(x)(e^{-\pi x^2} - 1) \geq \int_C \phi(x)(e^{-\pi x^2} - 1)
\]
So we see that
\begin{align*}
\int_{\R \setminus C} f(x)(e^{-\pi x^2})dx + \int_C f(x)(e^{-\pi x^2} - 1)dx
\geq& \int_{\R \setminus C} \phi(x)(e^{-\pi x^2})dx + \int_C f(x)(e^{-\pi x^2} - 1)dx \\
=& \int_\R \phi(x)(e^{-\pi x^2} - I[x \in C])dx = \E[\phi(X-U)] \\
=& \phi(\E[X-U]) = \phi(0) = 0  \text{.}
\end{align*}
Here the last equalities follow since $\phi$ is linear and both $X$ and $U$ have mean $0$. The base
case is thus proven.

We now assume that the claim is true for $n \geq 1$ and prove it for $n+1$. Note that $X = (X_1,\dots,X_{n+1})$
where the $X_i$s are i.i.d. gaussians with density $e^{-\pi x^2}$, and that $U = (U_1,\dots,U_{n+1})$
where the $U_i$s are i.i.d. uniform random variables on $C$. We first show that
\[
\E[f(X_1,\dots,X_{n+1})] \geq \E[f(X_1,\dots,X_n,U_{n+1})]
\]
To see this, note that
\[
\E[f(X_1,\dots,X_{n+1})] = \int_{\R^n} e^{-\pi(\sum_{i=1}^n x_i^2)} \int_\R f(x_1,\dots,x_{n+1}) e^{-\pi x_{n+1}^2}
dx_{n+1} \dots dx_1
\]
Now by convexity of $f$, we see that for any $x_1,\dots,x_n \in \R^n$ the function $g(y) = f(x_1,\dots,x_n, y)$ is a
convex function from $\R$ to $\R$. Therefore, by the analysis of the base case, we have that
\begin{align*}
\int_{\R^n} e^{-\pi(\sum_{i=1}^n x_i^2)} &\int_\R f(x_1,\dots,x_{n+1}) e^{-\pi x_{n+1}^2} dx_{n+1} \dots dx_1 \\
&\geq \int_{\R^n} e^{-\pi(\sum_{i=1}^n x_i^2)} \int_\R f(x_1,\dots,x_{n+1}) I[x_{n+1} \in C] dx_{n+1} \dots dx_1 \\
&= \E[f(X_1,\dots,X_n,U_{n+1})]
\end{align*}
as needed. Next by convexity of $f$, we get that the function
\[
g(x_1,\dots,x_n) = \E[f(x_1,\dots,x_n,U_{n+1})]
\]
is also convex. Therefore by the induction hypothesis, we get that
\[
\E[f(X_1,\dots,X_n,U_{n+1})] = \E[g(X_1,\dots,X_n)] \geq \E[g(U_1,\dots,U_n)] = \E[f(U_1,\dots,U_{n+1})]
\]
as needed.
\end{proof}

We are now ready for the main theorem of this section.

\begin{theorem}
\label{thm:l-norm-comp}
Let $s = \frac{1}{\sqrt{2\pi}} \sqrt{\frac{\log(2(2n+1))}{\pi}}$ and $C_s = \frac{1}{2s}[-1,1]^n$.
Define $D \subseteq \R^n$ as
\[
D = \left(\frac{1}{s} \Z^n\right) \bigcap \left(C_s + 2\sqrt{n}B_2^n\right)
\quad \text{ and } \quad p_x = \int_{C_s} \gamma_n(x + y)dy
\]
for $x \in D$. Then for any symmetric convex body $K \subseteq \R^n$, we have that
\[
\left(1-\frac{1}{s}\right) l(K) \leq \tilde{l}(K) \leq \left(1+\frac{1}{s}\right) l(K)
\]
where $\tilde{l}(K) = \sum_{x \in D} p_x \|x\|_K$.
\end{theorem}
\begin{proof}
The proof proceeds as follows. First we note in Claim 1 below that we can restrict attention to a ball of radius
$2\sqrt{n}$ via a tail bound on the standard Gaussian. Then, in Claim 2, we bound the error of the discrete approximation 
computed in terms of the norm of a random point from $U$ (uniform in $[-1/2, 1/2]^n$). 
Finally, using Lemma \ref{lem:unif-gaus-comp}, we can bound this norm by the $\ell$-estimate itself (Claim 3 below).

{\bf Claim 1.} 
\[
(1-e^{-0.3n}) \int_{\R^n} \|x\|_K \gamma_n(x) dx \leq \int_{D+C_s} \|x\|_K \gamma_n(x) dx
                                                \leq \int_{\R^n} \|x\|_K \gamma_n(x) dx
\]

{\bf Claim 2.}
\[
\left|\sum_{x \in D} p_x \|x\|_K - \int_{C+D_s} \|x\|_K \gamma_n(x)dx\right| \leq \frac{2}{s} \E[\|U\|_K].
\]

{\bf Claim 3.}
\[
\E[\|U\|_K] \leq \frac{1}{\sqrt{2\pi}} \E[\|X\|_K]
\]
where $X$ is a standard $n$-dimensional Gaussian.

We prove these claims presently. 

Combining Claims $(1),(2),$ and $(3)$, we get the upper bound
\begin{align*}
\sum_{x \in D} p_x \|x\|_K
        &\leq \int_{D+C_s} \|x\|_K \gamma_n(x)dx + \frac{2}{s}\E[\|U\|_K] \\
        &\leq \E[\|X\|_K] + \frac{\sqrt{2}}{\sqrt{\pi} s}\E[\|X\|_K] = \left(1+\frac{\sqrt{2}}{\sqrt{\pi}s}\right)\E[\|X\|_K]
        \text{, }
\end{align*}
and the lower bound
\begin{align*}
\sum_{x \in D} p_x \|x\|_K &\geq \int_{D+C_s} \|x\|_K \gamma_n(x) dx - \frac{2}{s}\E[\|U\|_K] \\
                           &\geq \left(1-e^{-0.3n}\right)\E[\|X\|_K] - \frac{\sqrt{2}}{\sqrt{\pi}s}\E[\|X\|_K]
                           = \left(1-e^{-0.3n}-\frac{\sqrt{2}}{\sqrt{\pi}s}\right)\E[\|X\|_K] \text{.}
\end{align*}
Since $e^{-0.3n}+\frac{\sqrt{2}}{\sqrt{\pi}s} \leq \frac{1}{s}$ for $n$ large enough, we get the claimed result.

Now we prove the claims. 
\paragraph{Proof of Claim 1:} Since the scaled cube $C_s$ tiles space with respect to the lattice $\frac{1}{s} \Z^n$,
we get by construction of $D$ that $2\sqrt{n}B_2^n \subseteq D+C_s$. Since $\|\cdot\|_K$ is non-negative, we
clearly have that
\[
\int_{2\sqrt{n}B_2^n} \|x\|_K \gamma_n(x)dx \leq \int_{D+C_s} \|x\|_K \gamma_n(x)dx \leq 
\int_{\R^n} \|x\|_K \gamma_n(x)dx
\]
Expressing the integral in polar coordinates, we have
\begin{align*}
\int_{2\sqrt{n}B_2^n} \|x\|_K \gamma_n(x) \, dx &= 
 \left(\frac{1}{\sqrt{2\pi}}\right)^n\int_{S^{n-1}} \int_0^{2\sqrt{n}} \|r\theta\|_K e^{-\frac{1}{2}r^2}r^{n-1} dr d\theta\\
                              &= \left(\frac{1}{\sqrt{2\pi}}\right)^n
                    \int_{S^{n-1}} \int_0^{2\sqrt{n}} \|\theta\|_K e^{-\frac{1}{2}r^2} r^n dr d\theta.
\end{align*}
Thus, 
\begin{align*}
\frac{\int_{2\sqrt{n}B_2^n} \|x\|_K \gamma_n(x) \, dx}{\int_{\R^n} \|x\|_K \gamma_n(x) \, dx} &=
\frac{\int_{S^{n-1}} \int_0^{2\sqrt{n}} \|\theta\|_K e^{-\frac{1}{2}r^2} r^n dr d\theta}{\int_{S^{n-1}} \int_0^{\infty} \|\theta\|_K e^{-\frac{1}{2}r^2} r^n dr d\theta}\\
&\ge \frac{\int_0^{2\sqrt{n}} e^{-\frac{1}{2}r^2} r^n dr}{\int_0^{\infty} e^{-\frac{1}{2}r^2} r^n dr}\\
&= 1 - \int_{\R^{n+1} \setminus 2\sqrt{n}B_2^n} \gamma_{n+1}(x) dx \\
&\geq 1 - e^{-(1-\frac{1+\ln(\frac{4n}{n+1})}{\frac{4n}{n+1}})2n} \geq 1 - e^{-0.3n}
\end{align*}
using Lemma \ref{lem:gaussian-tail} (i.e., the standard Gaussian tailbound)
with $t = 2\sqrt{\frac{n}{n+1}}$, and noting that $n \geq 1$.
This proves the claim.

\paragraph{Proof of Claim 2:}
For $y \in \R^n$, let $r(y)$ denote the closest vector to $y$ in
$\frac{1}{s}\Z^n$ under the $l_2$ norm. Given the structure of $\Z^n$, a simple
computation yields that
\[
r(y) = \left(\frac{\round{sy_1}}{s}, \dots, \frac{\round{sy_n}}{s}\right)
\]
Furthermore, for $x \in \frac{1}{s}\Z^n$ we have that $r(y) = x$ iff $y \in x+C_s$. Now we see that
\[
\sum_{x \in D} p_x \|x\|_K = \sum_{x \in D} \int_{x+C_s} \|x\|_K \gamma_n(y) dy
                           = \int_{D+C_s} \|r(y)\|_K \gamma_n(y) dy
\]
From here, using the triangle inequality, we get that
\[
\int_{D+C_s} \|r(y)\|_K \gamma_n(y) dy \leq \int_{D+C_s} (\|y\|_K+\|y-r(y)\|_K) \gamma_n(y) dy
                                       = \int_{D+C_s} \|y\|_K \gamma_n(y) + \int_{D+C_s} \|y-r(y)\|_K \gamma_n(y) dy
\]
Similarly, we also get that
\[
\int_{D+C_s} \|r(y)\|_K \gamma_n(y) dy \geq \int_{D+C_s} \|y\|_K \gamma_n(y) - \int_{D+C_s} \|y-r(y)\|_K \gamma_n(y) dy
\]
Hence to get the desired upper and lower bounds on $\sum_{x \in D} p_x \|x\|_K$, we need only upper bound
the quantity $\int_{D+C_s} \|y-r(y)\|_K \gamma_n(y) dy$. Now we note that
\begin{align*}
\int_{D+C_s} \|y-r(y)\|_K \gamma_n(y) dy &= \int_{C_s} \|c\|_K \sum_{y \in D + c} \gamma_n(y) dc \\
   &= \left(\frac{1}{s}\right)^n \int_{C_1} \left\|\frac{c}{s}\right\|_K
           \sum_{y \in D + \frac{c}{s}} \gamma_n(y) dc \\
   &= \left(\frac{1}{s}\right)^n \int_{C_1} \left\|\frac{c}{s}\right\|_K
           \sum_{y \in sD + c} \gamma_n\left(\frac{y}{s}\right) dc \\
   &= \left(\frac{1}{\sqrt{2\pi}s}\right)^{n} ~ \frac{1}{s} ~ \int_{C_1} \|c\|_K
           \sum_{y \in sD + c} e^{-\pi\|\frac{y}{\sqrt{2\pi} s}\|^2} dc
\end{align*}
Next note that $sD = \Z^n \cap \left(C_1 + 2\sqrt{n}sB_2^n\right)$. Therefore by Lemma \ref{lem:zn-ineqs} we have that
\begin{align*}
\left(\frac{1}{\sqrt{2\pi}s}\right)^{n} ~ \frac{1}{s} ~ \int_{C_1} \|c\|_K
          \sum_{y \in sD + c} e^{-\pi\|\frac{y}{\sqrt{2\pi} s}\|^2} dc
\leq& ~ \left(\frac{1}{\sqrt{2\pi}s}\right)^{n} ~ \frac{1}{s} ~ \int_{C_1} \|c\|_K
           \sum_{y \in \Z^n + c} e^{-\pi\|\frac{y}{\sqrt{2\pi} s}\|^2} dc \\
\leq& ~ \left(\frac{1}{\sqrt{2\pi}s}\right)^{n} ~ \frac{1}{s} ~ \int_{C_1} \|c\|_K (\sqrt{2\pi} s)^n (1+\frac{1}{2n})^n dc \\
\leq& ~ \frac{2}{s} \int_{C_1} \|c\|_K dc = \frac{2}{s} \E[\|U\|_K]
\end{align*}

\paragraph{Proof of Claim 3:} We wish to show that
\[
\E[\|U\|_K] \leq \frac{1}{\sqrt{2\pi}} \E[\|X\|_K] = \E[\|\frac{1}{\sqrt{2\pi}}X\|_K]
\]
A simple computation gives that $\frac{1}{\sqrt{2\pi}}X$ has density $e^{-\pi\|x\|^2}$ for $x \in \R^n$.
Since $\|\cdot\|_K$ is a convex function, the above inquality follows directly from Lemma \ref{lem:unif-gaus-comp}.
The claim thus follows.
\end{proof}

\subsection{Efficiency of solving the convex program}

In what follows we will assume that our symmetric convex body $K$ is well sandwiched, i.e. that $B_2^n \subseteq K
\subseteq nB_2^n$. As mentioned previously, this can be achieved by GLS type rounding using the ellipsoid algorithm.

We recall the functions $f,\tilde{f}:\R^{n \times n} \rightarrow \R$
\[
f(A) = \int_{\R^n} \|Ax\|_K \gamma_n(x) dx \quad \text{ and } \quad \tilde{f}(A) = \sum_{x \in D} p_x \|Ax\|_K
\]
We will consider an approximate version of Program \ref{ell-prog}:
\begin{align}
\label{ell-prog-approx}
\begin{split}
\inf \tilde{f}(A) &= \sum_{x \in D} p_x \|Ax\|_{K} \\
\text{subject to} & \\
	          & \quad A \succeq 0 \\
                  & \quad \det(A) \geq 1
\end{split}
\end{align}

The main result of this section is the following:

\begin{theorem} Let $\tilde{A}$ denote an optimal solution to Program \ref{ell-prog-approx}. Then for $0 < \eps \leq 1$, a
matrix $A \in \R^{n \times n}$ satisfying $\tilde{f}(A) \leq (1+\eps)\tilde{f}(\tilde{A})$ can be computed in
deterministic $\poly(n, \ln \frac{1}{\eps})O(\sqrt{\log n})^n$ time. Furthermore, let $A \in \R^{n \times n}$
be any $2$-approximate solution to \ref{ell-prog-approx}, then for $E = \frac{\sqrt{n}}{\tilde{f}(A)} AB_2^n$ we have that
\[
N(E,K) = 2^{O(n)} \quad N(K,E) = O(\log n)^n
\]
\end{theorem}
\begin{proof}
Let $A^*$ denote an optimal solution to \ref{ell-prog}. Then by Theorem \ref{thm:l-norm-comp} we have that
\begin{equation}
\label{eq:approx-equiv}
(1-\frac{1}{s})f(A^*) \leq (1-\frac{1}{s})f(\tilde{A}) \leq \tilde{f}(\tilde{A}) 
                      \leq \tilde{f}(A^*) \leq (1+\frac{1}{s})f(A^*)
\end{equation}
where the first inequality follows by optimality of $A^*$, and the third inequality by optimality of $\tilde{A}$.

\paragraph{Claim 1:} $f(A^*) = O\left(\frac{\log n}{\vol(K)^{\frac{1}{n}}}\right)$.
Pick a linear transformation $T \in SL(n)$ minimizing $l(TK)l(T^*K)$. From the proof of Lemma \ref{thm:l-pos-cov},
for some $c_1,c_2 = \Theta(1)$, letting $m = c_1 ~ \frac{\sqrt{n}}{l(TK)}$ we have that
\[
\frac{1}{2} \vol(mT^{-1}B_2^n) = \vol(mT^{-1}B_2^n \cap K) \geq \left(\frac{c_2}{\log n}\right)^n \vol(K)
\]
Now 
\[
\vol(mT^{-1}B_2^n) = \vol(B_2^n) \det(T^{-1}) m^n = \vol(B_2^n) \frac{m^n}{\det(T)} = \vol(B_2^n) m^n \text{.}
\]
Therefore
\[
\vol(B_2^n)^{\frac{1}{n}} m \geq \frac{c_2}{\log n} ~ \vol(K)^{\frac{1}{n}}
\Rightarrow \frac{c_1}{c_2} ~ \vol(B_2^n)^{\frac{1}{n}} ~ \sqrt{n} ~ \frac{\log n}{\vol(K)^{\frac{1}{n}}} \geq l(TK) 
\Rightarrow l(TK) = O\left(\frac{\log n}{\vol(K)^{\frac{1}{n}}}\right)
\]
Using the identity $\|x\|_{TK} = \|T^{-1}x\|_{K}$ we see that
\[
l(TK) = \int_{x \in \R^n} \|x\|_{TK} \gamma_n(x) dx = \int_{x \in \R^n} \|T^{-1}x\|_{K} \gamma_n(x) = f(T^{-1})
\]
Let $A = T^{-1}$. For a standard gaussian vector $X$ is $\R^n$, we note that $A_s = (A^tA)^{\frac{1}{2}}X$, where $A_s$
is the unique positive definite square root of $A^tA$, is identically distributed to $AX$.  Therefore $f(A_s) =
E[\|A_sX\|_K] = E[\|AX\|_K] = f(A) = f(T^{-1})$. Since $A_s = (A^tA)^{\frac{1}{2}} \succeq 0$ and
$\det(A_s) = |\det(A)| = \det(T^{-1}) = 1$, we have that $A_s$ is feasible for Program \ref{ell-prog}.
Since $A^*$ is the optimal solution to \ref{ell-prog} we have that
\[
f(A^*) \leq f(A_s) = f(T^{-1}) = O\left(\frac{\log n}{\vol(K)^{\frac{1}{n}}}\right)
\]
as needed.

\paragraph{Claim 2:} The Programs \ref{ell-prog} and \ref{ell-prog-approx} are convex.

By Lemma \ref{lem:lip-conv}, we know that both $f$ and $\tilde{f}$ are convex over the feasible region.
In both programs, the feasible region is the set of positive semi-definite matrices of determinant greater
than $1$, which is clearly convex.

\paragraph{Claim 3:} Program \ref{ell-prog-approx} can be solved to within $(1+\eps)$ multiplicative error in
deterministic \linebreak $\poly(n, \ln \frac{1}{\eps}) O(\sqrt{\log n})^n$ time.

Given that $B_2^n \subseteq K \subseteq nB_2^n$, by Lemma \ref{lem:region-bounds} we may constrain convex Program
\ref{ell-prog-approx} to the well-bounded region $R$ without removing any optimal solutions. Now by Lemma \ref{lem:lip-conv}
(3) the objective function is $2\sqrt{n}$ Lipshitz over operator norm (and hence over the Frobenius norm), and by
Lemma \ref{lem:region-bounds} (3) that the ratio of min and max value of the objective function over $R$ is
$O(n^\frac{5}{2})$. Given all this, we may apply the ellipsoid algorithm (see \cite{GLS} Theorem 4.3.13 for example)
\dnote{Put this in prelims maybe?} to solve the convex program \ref{ell-prog-approx} to within $(1+\eps)$ multiplicative
error using at most $\poly(n, \ln \frac{1}{\eps})$ evaluations of $\tilde{f}$ and arithmetic operations. Since each
evaluation of $\tilde{f}$ can be computed in deterministic $O(\sqrt{\log n})^n$ time, this proves the claim.

\paragraph{Claim 4:} Let $A$ be a $2$-approximation for the program \ref{ell-prog-approx}. Then the ellipsoid
$E = \frac{\sqrt{n}}{\tilde{f}(A)} AB_2^n$ satisfies $N(K,E) = O(\log n)^n$ and $N(E,K) = 2^{O(n)}$.

Let $\tilde{A}$ be as above. By Equation \eqref{eq:approx-equiv}, Lemma
\ref{thm:l-norm-comp} and Claim 1, we have that
\begin{align*}
f(A) &\leq \frac{s}{s-1} \tilde{f}(A) \leq 2 ~\frac{s}{s-1}~ \tilde{f}(\tilde{A}) 
     \leq \frac{s+1}{s-1} f(A^*) = O\left(\frac{\log n}{\vol(K)^{\frac{1}{n}}}\right) \text{.}
\end{align*}
By Theorem \ref{thm:l-norm-comp}, we note that $\frac{\sqrt{n}}{\tilde{f}(A)} = \Theta(1) \frac{\sqrt{n}}{f(A)}$. Hence
by Lemma \ref{lem:l-pos-ess}, there exists $c \leq 1$, where $c = \Omega(1)$, such that $\vol(cE \cap K) = \frac{1}{2}
\vol(cE)$. Now note that
\begin{align*}
\vol(cE) &= \left(\frac{c \sqrt{n}}{\tilde{f}(A)}\right)^n \det(A) \vol(B_2^n) 
         \geq \left(\frac{c \sqrt{n} \vol(B_2^n)^{\frac{1}{n}}}{\tilde{f}(A)}\right)^n 
         = \Omega\left(\frac{1}{\log n}\right)^n \vol(K) 
\end{align*}
Now since $\vol(E \cap K) \geq \vol(cE \cap K) = \frac{1}{2} \vol(cE) = \frac{1}{2} c^n \vol(E)$ and $\vol(E \cap K)
\geq \vol(cE \cap K) = \Omega\left(\frac{1}{\log n}\right)^n \vol(K)$, applying the covering estimates of Lemma
\ref{lem:cov-est} yields the claim.

\end{proof}

\begin{lemma}\label{lem:lip-conv}\hspace{1em}
\begin{enumerate}
\item $f$,$\tilde{f}$ define norms on $\R^{n \times n}$.
\item $A^tA \succeq B^tB \Rightarrow f(A) \geq f(B)$. 
\item $|f(A)-f(B)|,|\tilde{f}(A)-\tilde{f}(B)| \leq 2\sqrt{n}\|A-B\|$, where $\|A-B\|$ denote the operator norm of $A-B$.
\end{enumerate}
\end{lemma}
\begin{proof}
Let $X \in \R^n$ denote a standard Gaussian random vector. Take $A,B \in \R^{n \times n}$ and scalars $s,t \in \R$. 
Then note that
\[
f(sA + tB) = \E[\|(sA + tB)X\|_K] = \E[\|sAX + tBX\|_K] \leq \E[|s|\|AX\|_K + |t|\|BX\|_K]
           = |s|f(A) + |t|f(B)
\]
where the inequality above follows since $\|\cdot\|_K$ defines a norm. Lastly, using the fact that 
\[
\frac{1}{n}\|x\|_2\leq \|x\|_K \leq \|x\|_2
\] 
for $x \in \R^n$ (since $B_2^n \subseteq K \subseteq nB_2^n$) it is easy to verify that $f(A)
= 0 \Leftrightarrow A = 0^{n \times n}$ and $f(A) < \infty$ for all $A \in \R^{n \times n}$. Hence $f$ defines a norm on
$\R^{n \times n}$ as claimed. The argument for $\tilde{f}$ is symmetric.

Now take $A,B$ satisfying the condition of (2). Note that $AX$ is an origin centered gaussian with covariance matrix
$\E[AX(AX)^t] = \E[AXX^tA^t] = A^tA$. Similarly $BX$ is origin centered with covariance $B^tB$. From our assumptions,
the matrix $C = A^tA-B^tB \succeq 0$, hence $C$ has a PSD square root which we denote $C^{\frac{1}{2}}$. Now let
$Y$ denote standard $n$-dimensional Gaussian independent from $X$. Now note that $BX + C^{\frac{1}{2}}Y$
is again a Gaussian vector with covariance $B^tB + C = A^tA$. Hence $BX + C^{\frac{1}{2}}Y$ is identically
distributed to $AX$. Therefore we see that
\begin{align*}
f(A) &= \E[\|AX\|_K] = \E[\|BX+C^{\frac{1}{2}}Y\|_K] = \E_X[\E_Y[\|BX+C^{\frac{1}{2}}Y\|_K]] \\
     &\geq \E_X[\|BX+C^{\frac{1}{2}}\E_Y[Y]\|_K] = \E[\|BX\|_K] = f(B)
\end{align*}
where the inequality follows by Jensen's inequality and the convexity of $\|\cdot\|_K$.

We now prove (3). Take $A,B \in \R^{n \times n}$. By the triangle inequality, we have that
\[
f(B) - f(A-B) \leq f(A) \leq f(B) + f(A-B) \text{.}
\]
Therefore $|f(B)-f(A)| \leq f(A-B)$. Since $\tilde{f}$ is also a norm, we similarly get that
$|\tilde{f}(B)-\tilde{f}(A)| \leq \tilde{f}(A-B)$. Let $\lambda = \|A-B\|$. By definition of the
operator norm, we have that $(A-B)^t(A-B) \preceq \lambda^2 I_n$, where $I_n$ denote the $n \times n$
identity matrix. Therefore by $(2)$, we have that
\begin{align*}
f(A-B) &= \E[\|(A-B)X\|_K] \leq \E[\|\lambda X\|_K] = \lambda \E[\|X\|_K] \\
       &\leq \lambda \E[\|X\|_2] \leq \lambda \sqrt{\E[\|X\|_2^2]} = \lambda \sqrt{n}
\end{align*}
as needed. Next by Theorem \ref{thm:l-norm-comp}, we have that
\[
\tilde{f}(A-B) \leq 2f(A-B) \leq 2\|A-B\|\sqrt{n}
\]
as required.
\end{proof}

\begin{lemma}\label{lem:region-bounds}
Define the set
\[
R = \set{A \in \R^{n \times n}: A \succeq 0, \det(A) \geq 1, \|A\| \leq 2n^{\nicefrac{3}{2}}} 
\]
where $\|A\|$ denote the operator norm of $A$. Then $R$ satisfies the following:
\begin{enumerate}
\item $R$ contains an optimal solution to the programs \ref{ell-prog} and \ref{ell-prog-approx}.
\item $R$ satisfies the following sandwiching properties:
\[
n^{\frac{3}{2}} I_n + (n^{\frac{3}{2}}-1)B_2^{n \times n} \subseteq R \subseteq 
n^{\frac{3}{2}} I_n + 3n^2 B_2^{n \times n}
\]
where $I_n$ is the $n \times n$ identity matrix and $B_2^{n \times n} = \set{A \in \R^{n \times n}: A = A^t, \|A\|_F
\leq 1}$, the set of $n \times n$ symmetric matrices of Frobenius norm at most $1$.  
\item There is an absolute constant $c$ such that for any $A \in R$, we have that
\[
\frac{c}{\sqrt{n}} \leq f(A),\tilde{f}(A) \leq 3n^2
\]
\end{enumerate}
\end{lemma}

\begin{proof}
Let $X \in \R^n$ denote a standard $n$ dimensional gaussian vector, and let $s = \frac{1}{\sqrt{2\pi}}
\sqrt{\frac{\log(2(2n+1))}{\pi}}$.

We start by showing property (1). Let $A$ be an optimal solution for Program \ref{ell-prog}. We wish to show that $\|A\| \leq
n^{\frac{3}{2}}$. Since $\|x\|_2 \geq \|x\|_K$ for all $x \in \R^n$, we have that
\[
f(I_n) = \E[\|X\|_K] \leq \E[\|X\|_2] \leq \sqrt{E[\|X\|_2^2} = \sqrt{n}
\]
Since $I_n$ is feasible for \ref{ell-prog}, it suffices to show that if $\|A\| \geq 2n^{\frac{3}{2}}$, we get that
$f(A) \geq \sqrt{n}$. Let $\lambda = \|A\|$, and let $v$ denote an eigenvector of $A$ satisfying $Av = \lambda
v$ and $\|v\|=\frac{1}{n}$.  Since $K \subseteq nB_2^n$, we have that $K \subseteq W = \set{x: |\pr{v}{x}| \leq 1}$
(since $\|v\|=\frac{1}{n}$). Therefore 
\begin{align*}
f(A) &= \E[\|X\|_K] \geq \E[\|AX\|_W] = \E[|\pr{v}{AX}|] \\
     &= \lambda \E[|\pr{v}{X}|] = \lambda(\sqrt{\frac{2}{\pi}}\|v\|) = \frac{\lambda}{n} \sqrt{\frac{2}{\pi}}
\end{align*}
Since $A$ is optimal, we get that $\frac{\lambda}{n} \sqrt{\frac{2}{\pi}} \leq \sqrt{n} \Rightarrow \lambda \leq
2n^{\frac{3}{2}}$ as claimed. We now show the same for Program \ref{ell-prog-approx}. By \ref{thm:l-norm-comp},
$\tilde{f}(I_n) \leq (1+\frac{1}{s})f(I_n) \leq (1+\frac{1}{s})\sqrt{n}$. Now if $A$ is an optimal solution
to \ref{ell-prog-approx}, letting $\lambda = \|A\|$, we have that
\[
\tilde{f}(A) \geq (1-\frac{1}{s})f(A) \geq (1-\frac{1}{s})\frac{\lambda}{n} \sqrt{\frac{2}{\pi}}
\]
But then as above we have that
\[
\lambda \leq \frac{1+\frac{1}{s}}{1-\frac{1}{s}} ~ \sqrt{\frac{\pi}{2}} ~ n^{\frac{3}{2}} \leq 2 n^{\frac{3}{2}}
\]
for $n$ large enough as needed. Therefore $R$ satisfies property (1) as needed.

We now show the containment relationship in (2). Take $A = n^{\frac{3}{2}}I_n + B$ where $B \in (n^{\frac{3}{2}}-1)B_2^{n
\times n}$. We must show that $A \in R$. We recall that $\|B\| \leq \|B\|_F \leq \sqrt{n} \|B\|$. First, note that 
\[
\|A\| \leq n^{\frac{3}{2}} + \|B\| \leq n^{\frac{3}{2}} + n^{\frac{3}{2}} - 1 < 2n^{\frac{3}{2}}
\]
as needed. Next note that
\begin{align*}
\inf_{v \in S^{n-1}} v^tAv &=  \inf_{v \in S^{n-1}} v^t(n^\frac{3}{2}I_n + B)v
                            \geq \inf_{v \in S^{n-1}} n^{\frac{3}{2}}v^t v - v^tBv \\
                            &= n^{\frac{3}{2}} - \sup_{v \in S^{n-1}} v^tBv \geq n^{\frac{3}{2}} - \|B\| \geq 1
\end{align*}
Since $A$ is symmetric, the above shows the $A$'s smallest eigenvalue is at least $1$, and hence $A \succeq 0$ and $\det(A)
\geq 1$ as needed. To show the opposite containment, note that for $A \in R$, we have that
\[
\|A-n^{\frac{3}{2}}I_n\|_F \leq \|A\|_F + \|n^{\frac{3}{2}}I_n\|_F \leq \sqrt{n}\|A\| + n^2 \leq 3n^2
\]
as needed.

Now we need to show the bounds on $f(A)$ for $A \in R$ to prove property (3). First we remember that
\[
\E[\|AX\|_2] \geq f(A) \geq \frac{1}{n} \E[\|AX\|_2]
\]
Hence it suffices to upper and lower bound $\E[\|AX\|_2]$. We see that
\[
c \sqrt{\E[\|AX\|^2_2]} \leq \E[\|AX\|_2] \leq \sqrt{\E[\|AX\|^2_2]}
\]
for an absolute constant $0 \leq c < 1$. Here the first inequality follows by Borell's Lemma and the second by Jensen's
inequality. Next we have that
\[
\sqrt{\E[\|AX\|^2_2]} = \sqrt{\E[X^tA^tAX]} = \sqrt{\E[\mathrm{trace}(A^tAXX^t)]} 
                      = \sqrt{\mathrm{trace}(A^tA)} = \|A\|_F
\]
Since $A \in R$, we know that $\|A\| \leq 2n^{\frac{3}{2}}$, and hence $\|A\|_F \leq 2n^2$. Combining the
above inequalities, this yields that $f(A) \leq 2n^2$ as needed. We now prove the lower bound.
Since $A \in R$, we have that $\det(A) \geq 1$. Let $A_i$ denote the $i^{th}$ column of $A$. Now we have that
\[
\|A\|_F \geq \sqrt{n}\prod_{i=1}^n \|A_i\|_2^{\frac{1}{n}} \geq \sqrt{n} \det(A)^{\frac{1}{n}} \geq \sqrt{n}
\]
where the first inequality follows by the arithmetic - geometric mean inequality, and the second follows from
Hadamard's inequality. Combining the above inequalities, we get that
\[
f(A) \geq \frac{1}{n} \E[\|AX\|_2] \geq \frac{c}{n} \|A\|_F \geq \frac{c}{\sqrt{n}}
\]
as needed. The bounds for $\tilde{f}(A)$ follow from the relationship $(1-\frac{1}{s})f(A) \leq f(A) \leq
(1+\frac{1}{s})f(A)$ (Theorem \ref{thm:l-norm-comp}).
\end{proof}




\section{Application to lattice algorithms}\label{sec:SVP}

We now apply our construction of $\ell$-type ellipsoids to lattice algorithms.  
Dadush et al \cite{DPV-SVP-11} gave algorithms for SVP in any norm, CVP in any norm and Integer Programming (IP). 
These algorithms were all based on the construction of an $M$-ellipsoid. Their core result can be 
stated as follows. 
For a lattice $L$ and convex body $K$ in $\R^n$, let $G(K,L)$ be the
largest number of lattice points contained in any translate of $K$, i.e., 
\begin{equation}
  \label{eq:G-K-L}
  G(K,L) = \max_{x \in \R^{n}} \abs{(K+x) \cap L}.
\end{equation}
\begin{theorem}\cite{DPV-SVP-11}
  Given any convex body $K \subseteq \R^n$ along with an $M$-ellipsoid $E$ of $K$ and any $n$-dimensional
  lattice $L \subseteq \R^n$, the set $K \cap L$ can be computed in
  deterministic time $G(K,L) \cdot 2^{O(n)}$.
\end{theorem}
They then proceeded to give a randomized construction of an $M$-ellipsoid. The necessary properties of the $M$-ellipsoid
$E$ are that the covering numbers $N(K,E)$ and $N(E,K)$ are both bounded by $2^{O(n)}$. In fact, the result of
\cite{DPV-SVP-11} can be stated more generally as follows. 
\begin{theorem}
  Given any convex body $K \subseteq \R^n$ along with an ellipsoid $E$ of $K$ and any $n$-dimensional
  lattice $L \subseteq \R^n$, the set $K \cap L$ can be computed in
  deterministic time $G(K,L) \cdot N(K,E)N(E,K) \cdot 2^{O(n)}$.
\end{theorem}

Furthemore, in \cite{DPV-SVP-11}, they only require an algorithm which builds an M-ellipsoid when $K$ is centrally
symmetric. This follows since one can show that an M-ellipsoid $E$ for $K-K$ (which is symmetric) is also an M-ellipsoid
for $K$ (of slightly worse quality). Hence from Theorem \ref{thm:M-ellipsoid-algorithm} and the bounds derived on
$N(K,E)$ and $N(E,K)$, we obtain a simple corollary. 
\begin{corollary}
  Given any convex body $K \subseteq \R^n$ and any $n$-dimensional
  lattice $L \subseteq \R^n$, the set $K \cap L$ can be computed in
  deterministic time $G(K,L) \cdot O(\log n)^n$.
\end{corollary}

This lattice point enumerator is the core of subsequent algorithms for SVP, CVP and IP in \cite{DPV-SVP-11}.  We obtain
similar conclusions with deterministic algorithms but with an overhead of $O(\log n)^n$.  The precise statement for SVP
is Theorem \ref{thm:det-SVP-anynorm}. For CVP the statement is as follows. 

\begin{theorem}
  There is a deterministic algorithm that, given any well-centered
  $n$-dimensional convex body $K$, solves CVP exactly on any $n$-dimensional
  lattice $L$ in the semi-norm $\length{\cdot}_{K}$ defined by $K$, in
  $(2+\gamma)^{O(n)}\cdot O(\log n)^n$ time and space, provided that the distance from
  the query point $x$ to $L$ is at most $\gamma$ times the length of
  the shortest nonzero vector of $L$ (under $\length{\cdot}_{K}$).
\end{theorem}

A central motivation for solving SVP in general norms is to improve the complexity of integer programming. The IP
algorithm directly uses the SVP algorithm. Moreover, in this case, the final complexity bound is already higher than
$O(\log n)^n$, so we simply get the IP complexity of \cite{DPV-SVP-11} with a deterministic algorithm. 

\begin{theorem}
  There exists a deterministic algorithm that, given a convex body $K
  \subseteq \R^{n}$ and an $n$-dimensional lattice $L \subset \R^{n}$,
  either decides that $K \cap L = \emptyset$ or returns a point $y \in
  K \cap L$ in expected $O(f^*(n))^n$ time, where $f^{*}(n)$ is the
  optimal bound for the ``flatness theorem.''
\end{theorem}

The flatness theorem, which we do not describe here, gives a bound on the lattice width of lattice-point-free convex bodies.

\section{Conclusion}

The $\ell$-ellipsoid with its covering guarantees is in fact the starting point of Milman and Bourgain's proof of the existence of $M$-ellipsoids.
However, unlike the $\ell$-ellipsoid, we are not aware of any convex programming formulation of $M$-ellipsoids. 

It remains open to give a deterministic $2^{O(n)}$ algorithm for $M$-ellipsoids and coverings. This would resolve the open
problem of a deterministic $2^{O(n)}$ SVP algorithm in any norm. 

Another open problem is to fully extend the approach suggested in \cite{DPV-SVP-11} to exact or $(1+\eps)$ CVP. At the
moment, their result only holds for exact CVP when the target point's distance to the lattice is at most a constant
times the minimum distance of the lattice. In particular, it is open to give a $2^{O(n)}$ algorithm for the CVP under
the $L_\infty$ norm. 

{\bf Acknowledgments. } We are deeply grateful to Grigoris Paouris and Chris Peikert for illuminating discussions, and
to Gilles Pisier for his book on convex bodies.  

\bibliographystyle{alphaabbrvprelim}
\bibliography{lattices,acg,cg}

\end{document}